\documentclass[letterpaper, 10 pt, conference]{ieeeconf}  

\IEEEoverridecommandlockouts                              

\usepackage{amsmath,amssymb,array,arydshln,color}
\usepackage{cuted}
\usepackage{soul}
\usepackage{flushend}
\usepackage{mathtools}
\mathtoolsset{showonlyrefs}
\usepackage[colorinlistoftodos, textwidth=2cm, shadow]{todonotes}

\usepackage{flushend}

\def\Ac{{\mathcal A}}

\def\Cbb{{\mathbb C}}

\def\Ec{{\mathcal E}}

\def\Gc{{\mathcal G}}

\def\Hc{{\mathcal H}}

\def\Ibb{{\mathbb I}}

\def\Jc{{\mathcal J}}

\def\Lc{{\mathcal L}}

\def\Nc{{\mathcal N}}

\def\Pbb{{\mathbb P}}

\def\Rc{{\mathcal R}}
\def\Rbb{{\mathbb R}}

\def\Vc{{\mathcal V}}

\def\0{{\bf 0}}

\newcommand{\bitem}{\begin{itemize}}
\newcommand{\eitem}{\end{itemize}}
\newcommand{\btabular}{\begin{tabular}}
\newcommand{\etabular}{\end{tabular}}
\newcommand{\bcenter}{\begin{center}}
\newcommand{\ecenter}{\end{center}}
\newcommand{\bea}{\begin{eqnarray}}
\newcommand{\eea}{\end{eqnarray}}
\newcommand{\bean}{\begin{eqnarray*}}
\newcommand{\eean}{\end{eqnarray*}}
\newcommand{\ba}{\left[ \begin{array}}
\newcommand{\ea}{\\ \end{array} \right]}
\newcommand{\bear}{\begin{array}}
\newcommand{\eear}{\\ \end{array}}

\newcommand{\non}{\nonumber}

\newcommand*{\QEDB}{\hfill\ensuremath{\blacksquare}}%
\newcommand*{\QET}{\hfill\ensuremath{\triangleleft}}
\newcommand{\norm}[1]{\left\lVert#1\right\rVert}

\newcounter{subequation}
\def\beasub{\addtocounter{equation}{+1}
\setcounter{subequation}{\value{equation}}
\setcounter{equation}{0}
\renewcommand{\theequation}{\arabic{subequation}\alph{equation}}
\begin{eqnarray}}
\def\eeasub{\end{eqnarray}
\setcounter{equation}{\value{subequation}}
\renewcommand{\theequation}{\arabic{equation}}}

\def\inf{\operatornamewithlimits{inf\vphantom{p}}}
\newtheorem{problem}{Problem}

\newtheorem{theorem}{Theorem}[section]
\newtheorem{definition}{Definition}[section]
\newtheorem{assumption}{Assumption}[section]

\newtheorem{lemma}[theorem]{Lemma}

\newtheorem{rem}{Remark}

\title{\LARGE \bf A Zero-Sum Game Framework for Optimal Sensor Placement in Uncertain Networked Control Systems under Cyber-Attacks
}

\author{Anh Tung Nguyen$^{1}$, Sribalaji C. Anand$^{2}$, and Andr\'e M. H. Teixeira$^{1}$
\thanks{*This work is supported by the Swedish Research Council under
the grants 2018-04396 and 2021-06316 and by the Swedish Foundation for Strategic Research.}
\thanks{$^{1}$ Anh Tung Nguyen and Andr\'e M. H. Teixeira are with the Department of Information Technology, Uppsala University, PO Box 337, SE-75105, Uppsala, Sweden. {\tt\small \{anh.tung.nguyen, andre.teixeira\}@it.uu.se}}
\thanks{$^{2}$ Sribalaji C. Anand is with the Department of Electrical Engineering, Uppsala University, PO Box 65, SE-75103, Uppsala, Sweden. {\tt\small sribalaji.anand@angstrom.uu.se}}%
}

\begin{document}

\maketitle
\thispagestyle{empty}
\pagestyle{empty}

\begin{abstract}
This paper proposes a game-theoretic approach to address the problem of optimal sensor placement against an adversary in uncertain networked control systems. 
The problem is formulated as a zero-sum game with two players, namely a malicious adversary and a detector. 
Given a protected performance vertex, we consider a detector, with uncertain system knowledge, that selects another vertex on which to place a sensor and monitors its output with the aim of detecting the presence of the adversary.
On the other hand, the adversary, also with uncertain system knowledge, chooses a single vertex and conducts a cyber-attack on its input.
The purpose of the adversary is to drive the attack vertex as to maximally disrupt the protected performance vertex while remaining undetected by the detector.
As our first contribution,
the game payoff of the above-defined zero-sum game is formulated in terms of the Value-at-Risk of the adversary's impact.
However, this game payoff corresponds to an intractable optimization problem.
To tackle the problem, we adopt the scenario approach to approximately compute the game payoff.
Then, the optimal monitor selection is determined by analyzing the  equilibrium of the zero-sum game.
The proposed approach is illustrated via a numerical example of a 10-vertex networked control system.
\end{abstract}

\section{Introduction}
Networked control systems have been playing a crucial role in modeling, analysis, and operation of real-world large-scale interconnected systems such as power systems, transportation networks, and water distribution  networks.
Those systems consist of multiple interconnected  subsystems which generally communicate with each other via insecure communication channels to share their information.
This insecure protocol may leave the networked control systems vulnerable to cyber-attacks such as denial-of-service and false-data injection attacks \cite{teixeira2015secure}, inflicting serious financial loss and civil damages. 
Reports on actual damages such as Stuxnet \cite{falliere2011w32} and Industroyer \cite{kshetri2017hacking} have described the catastrophic consequences of such cyber-attacks for an Iranian nuclear program and a Ukrainian power grid, respectively.
Motivated by the above observation, cyber-physical security has increasingly received much attention from
control society in recent years.


One of the most popular security metrics is
the game-theoretic approach that has been successfully applied to deal with the problem of robustness, security, and resilience of networked control systems \cite{zhu2015game}.
This approach affords us to address the robustness and security of networked control systems within the common well-defined framework  of $\Hc_\infty$ robust control design.
Further, many other concepts of games considering networked systems subjected to cyber-attacks such as dynamic  \cite{gupta2016dynamic}, stochastic  \cite{miao2018hybrid},
network monitoring  \cite{milovsevic2019network,pirani2021game}, and zero-sum games \cite{nguyen2022single}
have been recently studied.
Although the above games were successful in studying control systems subjected to cyber-attacks such as denial-of-service and stealthy data injection attacks, the full system model knowledge was assumed to be available to both the malicious adversary and the detector.
This assumption might be restrictive when it comes to large-scale interconnected systems which can consist of a huge number of subsystems.
This can be explained by a variety of facts such as $(i)$ limited availability of computational resources for modeling, $(ii)$ limited availability of modeling data, and $(iii)$ modeling errors.
Thus, the adversary and the detector might have  limited system knowledge instead of  accurate system parameters, which will be addressed throughout this paper.

In this paper, we deal with the problem of optimal sensor placement against an adversary in an uncertain networked control system which is represented by interconnected vertices.
Given a protected performance vertex, the detector monitors the system by selecting a single monitor vertex and placing a sensor to measure its output with the purpose of detecting cyber-attacks.
Meanwhile, the adversary chooses a single vertex to attack and directly injects attack signals into its input via the wireless network.
The aim of the adversary is to steer the attack vertex as to maximally disrupt the protected performance vertex while remaining undetected by the detector.
The contributions of this paper are the following

\begin{enumerate}
    \item 
    The problem of optimal sensor placement against the adversary is formulated as a zero-sum game between two strategic players, i.e., the adversary and the detector, with the same uncertain system knowledge.
    \item Due to the uncertainty, the game payoff of the zero-sum game, which is a min-max optimization problem,  is computationally intractable~\cite{anand2021risk}.  
    To deal with the problem, we adopt the scenario approach \cite{calafiore2007probabilistic} to approximately compute the above game payoff.
    \item We show that the existence of a finite solution to the problem is related to the system-theoretic properties of the  dynamical system, namely its invariant zeros and relative degrees.
    \item The solutions to the problem of the optimal sensor placement are provided by investigating the pure and the mixed-strategy equilibrium of the zero-sum game in a numerical example.
\end{enumerate}

\begin{figure}[!t]
	\centering
	\includegraphics[width=0.41\textwidth]{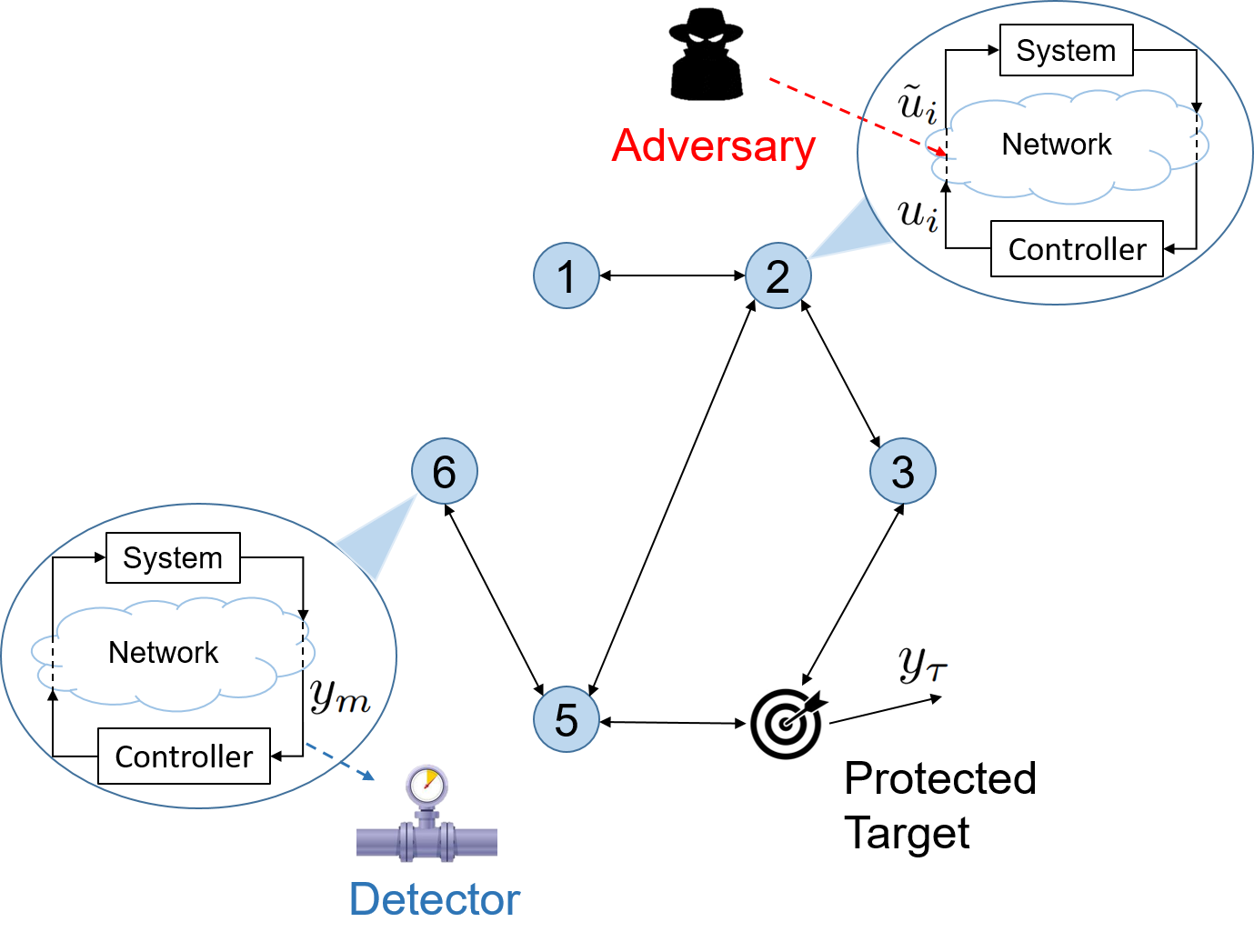}
	\caption{Visualization of a zero-sum game between a detector and an adversary in a networked control system.}
	\label{fig:illustration}
\end{figure}

We conclude this section by providing the notations which are used throughout this paper.
The problem description is given in Section \ref{sec:prob_back}.
Thereafter, Section \ref{sec:prob_form} formulates the problem of optimal sensor placement as a zero-sum game with the game payoff based on a risk metric.
The evaluation of the game payoff is carried out in Section \ref{sec:evl_game}.
Section~\ref{sec:num_eg} presents a numerical example of the zero-sum game between an adversary and a detector and computes the optimal monitor selection based on the mixed-strategy Nash equilibrium.
Concluding remarks are provided in Section \ref{sec:concl}.

{\bf Notation:} 
the set of real positive numbers is denoted as $\Rbb_+$; $\Rbb^n$ and $\Rbb^{n \times m}$ stand for sets of real $n$-dimensional vectors and $n$-row $m$-column matrices, respectively.
Let us define $e_i \in \Rbb^n$ with all zero elements except the $i$-th element that is set as $1$.
A continuous-time system with the state-space model $\dot{x}(t) = Ax(t) + Bu(t),\; y(t) = Cx(t) + Du(t)$ is denoted as $\Sigma \triangleq (A,B,C,D)$.
Consider the norm $\norm{x}_{\Lc_2 [0,T]}^2 \triangleq \int_{0}^{T} \norm{x(t)}_2^2~dt$.
The space of square-integrable  functions is defined as $\Lc_{2} \triangleq \bigl\{f: \Rbb_{+} \rightarrow \Rbb ~|~ \norm{f}_{\Lc_2 [0,\infty]} < \infty \bigr\} $
and the extended space is defined as $\Lc_{2e} \triangleq \bigl\{ f: \Rbb_{+} \rightarrow \Rbb ~|~ \norm{f}_{\Lc_2 [0,T]} < \infty,~ \forall~ 0 < T < \infty \bigr\} $.
We denote $\Ibb_\Ac(x)$ as an indicator function
such that $\Ibb_\Ac(x) = 1$ if $x \in \Ac$, otherwise $\Ibb_\Ac(x) = 0$.
The probability of $X$ is denoted as $\Pbb(X)$.
For $x \in \Rbb$, $\left \lceil{x}\right \rceil$ represents a value rounded to the nearest integer greater than or equal to $x$.
Let $\Gc \triangleq (\Vc, \Ec, A, \Theta)$ be an undirected weighted digraph with the set of $N$ vertices $\Vc = \{v_1, v_2,...,v_N\}$,
the set of edges $\Ec \subseteq \Vc \times \Vc $, the  weighted adjacency matrix $A \triangleq [a_{ij}]$, and the weighted self-loop matrix $\Theta$.
For any $(v_i,v_j) \in \Ec, ~i\neq j$, the element of the weighted matrix $a_{ij}$ is positive, 
and with $(v_i,v_j) \notin \Ec$ or $i = j$, $a_{ij} = 0$. 
The degree of vertex $v_i$ is denoted as 
$d_i =  \sum_{j=1}^{n} a_{ij}$ and the degree matrix of the graph $\Gc$ is defined as 
$D = {\bf diag}\big(d_1, d_2,\dots,d_N\big)$, where ${\bf diag}$ stands for a diagonal matrix.
For each vertex $v_i$, it has a positive weighted self-loop gain $\theta_i > 0$. The weighted self-loop matrix of the graph $\Gc$ is defined as $\Theta = {\bf diag}\big(\theta_1, \theta_2,\dots,\theta_N\big)$.
The Laplacian matrix, representing  the graph $\Gc$, is defined as $L = [\ell_{ij}] = D - A + \Theta$.
Further, $\Gc$ is called an undirected graph if  $A$ is symmetric.
An edge of an undirected graph $\Gc$ is denoted by a pair $(v_i,v_j) \in \Ec$. 
An undirected graph is connected if for any pair of vertices there exists at least one path between two vertices.
The set of all neighbours of vertex $v_i$ is denoted as $\Nc_i = \{v_j \in \Vc: (v_i,v_j) \in \Ec \}$.

\section{Problem description} 
\label{sec:prob_back}
This section firstly presents the description of a networked control system.
Then, we introduce 
a malicious adversary who with limited system knowledge conducts a cyber-attack to maliciously affect the system performance.
\subsection{Networked control system description}
Consider a networked control system associated with a connected undirected graph $\Gc \triangleq (\Vc, \Ec, A, \Theta)$ with $N$ vertices, the state-space model of each  one-dimensional vertex $v_i, ~i \in \bigl\{1,2,\ldots,N\bigr\}$, is described as
\begin{align}
	\dot x_i^\Delta(t) &=  \sum_{v_j \in \Nc_i} \ell_{ij}^\Delta \big(x_i^\Delta(t) - x_j^\Delta(t)\big) + \tilde{u}_i(t), 
	\label{sys:xi}
	\\
	y_\tau^\Delta(t) &= x^\Delta_\tau(t), \label{sys:xi_t}
\end{align}
where  $x_i^\Delta(t),\;\tilde{u}_i(t) \in \Rbb$ are the state of vertex $v_i$ and its control input received from its controller over the wireless network (see Fig. \ref{fig:illustration}), respectively. 
The performance of the networked control system \eqref{sys:xi} is measured via the state of a given vertex $v_\tau \in \Vc$ in \eqref{sys:xi_t}. 
The weight parameters $\ell_{ij}^\Delta, ~ \forall (v_i,v_j) \in \Ec,$ are uncertain and assumed to be structured as $\ell_{ij}^\Delta \triangleq \bar{\ell}_{ij} + \delta_{ij}$, where $\bar{\ell}_{ij}$ and $\delta_{ij}$ are the nominal value and the bounded probabilistic uncertainty of $\ell_{ij}^\Delta$, respectively.

First, we consider the wireless network healthy, i.e., the absence of cyber-attacks. 
Thus, the received control input $\tilde{u}_i(t)$ of vertex $v_i,~i \in \bigl\{1,2,\ldots,N\bigr\}$, is the same as the  control input  sent by its controller: 
\begin{align}
    \tilde{u}_i(t) = u_i(t) = -\theta_i x_i^\Delta(t),
	\label{sys:u}
\end{align}
where $u_i(t)$ is the control input designed and sent by the controller of vertex $v_i$. 
$\theta_i \in \Rbb_+$ is an adjustable self-loop control gain of vertex $v_i$.

For convenience, let us denote $x^\Delta(t) \triangleq \big[x_1^\Delta(t),~x_2^\Delta(t),\ldots,~x_N^\Delta(t)\big]^\top$ as the state of the networked control system.
The dynamics of the networked control system \eqref{sys:xi} under the control law \eqref{sys:u} can be rewritten as
\begin{align}
    \dot x^\Delta(t) &= - L^\Delta x^\Delta(t),
    \label{sys:x_noatt}
\end{align}
where the uncertain matrix $L^\Delta$ is defined as: $L^\Delta \triangleq \bar{L} + \Delta$, $\Delta  \in \Omega$,
    where $\Omega$ is a closed and bounded set,
    $\bar{L} \triangleq [\bar{\ell}_{ij}]$ and $\Delta \triangleq [\delta_{ij}]$ are nominal value and bounded uncertainty of $L^\Delta$, respectively.
Next,
let us make use of the following assumptions.
\begin{assumption} \label{ass:zero_init}
    We assume that the  healthy networked control system \eqref{sys:x_noatt} is at its equilibrium before being attacked.  
\end{assumption}
\begin{assumption}
    \label{ass:target_vertex}
    The input of the given performance vertex $v_\tau$ is protected from any attacks. 
    Further, its state is unmeasurable. 
\end{assumption}

Then, except for the protected target vertex $v_\tau$, a detector  monitors the system by placing a sensor at the output of a single vertex $v_m \in \Vc \setminus \{v_\tau\}$.
On the other hand, the system is attacked by an adversary, whose detailed descriptions are listed in the following subsection.


\subsection{Adversary description}
This part introduces resources and an attack strategy of the  adversary with limited system knowledge, so-called bounded-rational adversary \cite[Def. 2.2]{anand2021risk}.

\subsubsection{System knowledge}
The adversary knows the location of the protected target vertex $v_\tau$, the appearance of a detector, the set of $N$ vertices $\Vc$, and the set of edges $\Ec$.
However, the adversary does not know the exact location of the detector and has limited knowledge about $L^\Delta$ in  \eqref{sys:x_noatt}.
The adversary only knows  $\bar{L}$ and  $\Omega$ instead of $L^\Delta$.

\subsubsection{Disruption resource}
Except for the protected target vertex $v_\tau$, the adversary is able to conduct a cyber-attack on the input of another vertex.
The adversary firstly assumes the location of a monitor vertex $v_m$ selected by the detector.
Then, the adversary
selects a vertex $v_a \in \Vc \setminus \{v_\tau\}$ and injects a malicious attack signal $a(t) \in \Rbb$ on its input with the aim of manipulating the output of the protected target vertex $v_\tau$.
The control input \eqref{sys:u} of vertex $v_i, ~ i \in \{ 1,2,\ldots,N\}$, received from its controller over the attacked wireless network can be described as follows
\begin{align}
    \tilde{u}_i(t) = u_i(t) + 
    \begin{cases}
        0, ~ & v_i \neq v_a,
        \\
        a(t), ~ & v_i = v_a.
    \end{cases}
    \label{sys:ua}
\end{align}
Thus, the adversary perceives the system model \eqref{sys:x_noatt} under the control law \eqref{sys:ua}  with two  outputs at the two vertices $v_\tau$ and $v_m$ as an uncertain dynamical system described by
\begin{align}
	\dot x^\Delta(t) &= -L^\Delta x^\Delta(t) + e_a a(t),
	\label{sys:x_unc}
	\\
	y^\Delta_\tau (t) &= e_\tau^\top x^\Delta(t),
	\label{sys:yt}
	\\
	y^\Delta_m (t) &= e_m^\top x^\Delta(t).
	\label{sys:ym}
\end{align}

\subsubsection{Adversary strategy}
The goal of the adversary is to maliciously manipulate the output of the protected target vertex $v_\tau$ while remaining stealthy with the detector.
To this end, the adversary conducts the stealthy data injection attack, which is defined as follows.
Consider the above structure of the uncertain continuous-time system \eqref{sys:x_unc}-\eqref{sys:ym} which is denoted as $\Sigma^\Delta_{\tau,m} \triangleq (-L^\Delta,e_a,[e_\tau,e_m]^\top,0)$, with  target output $y_\tau^\Delta (t) = e_\tau^\top x^\Delta (t)$ and  monitor output $y^\Delta_m (t) = e_m^\top x^\Delta(t)$.
The input signal $a(t)$ of the system $\Sigma^\Delta_{\tau,m}$ is called the stealthy data injection attack if the  monitor output satisfies $\norm{y^\Delta_m}_{\Lc_{2}[0,T]}^2 < \sigma$, in which $\sigma > 0$ is called an alarm threshold. 
Further, the impact of the stealthy data injection attack is measured via the energy of the  target output over the horizon $[0,T]$, i.e., $\norm{y^\Delta_\tau}_{\Lc_{2}[0,T]}^2$.

Due to limited system knowledge, the uncertain system dynamics \eqref{sys:x_unc}-\eqref{sys:ym} are not explicitly available to the adversary.
Such an issue causes a difficulty for the adversary in designing of the attack strategy. 
To deal with the issue,
the next section adopts a risk metric to evaluate the attack impact over the probabilistic uncertainty set, which can be evaluated by the adversary to select an attack vertex.

\section{Problem formulation}
\label{sec:prob_form}

We consider that both the adversary and the detector have the same bounded uncertainty about the system knowledge. 
Based on this assumption, for a given uncertain parameter and attack and monitor vertices, the attack impact is characterized via an optimal control problem. Then, we aggregate the attack impact over the probabilistic uncertainty set by means of a risk metric. Finally, the  problem of optimal selection of attack and monitor vertices is formulated as a zero-sum game between two strategic players, the adversary and the detector, where the game payoff corresponds to the risk of the attack impact evaluated over the probabilistic uncertainty.
\subsection{Stealthy data injection attack policy}
Due to the presence of uncertainty in the system model \eqref{sys:x_unc}-\eqref{sys:ym}, the attack impact $J_\tau(v_a,v_m;\Delta,a)$ on the target vertex $v_\tau$ by the attack vector $a \in \Lc_{2e}$ becomes a function of the random variable  $\Delta \in \Omega$
\begin{align}
	J_\tau(v_a,v_m;\Delta,a) \triangleq 	\norm{y^\Delta_\tau}_{\Lc_2}^2 
	\Ibb_{\Ac} (a),
	\label{J_tau_delta}
	\\
	\Ac \triangleq \{ a |~  \norm{y^\Delta_m}_{\Lc_2}^2 \leq \sigma,\eqref{sys:x_unc},\eqref{sys:ym}, ~ x(0) = 0 \},
\end{align}
where 
$y^\Delta_\tau(t)$ and $y^\Delta_m(t)$ are the output of the target vertex $v_\tau$ and the output of the monitor vertex $v_m$, 
respectively.
From \eqref{J_tau_delta}, the worst-case attack impact on the target vertex $v_\tau$ with the random variable $\Delta \in \Omega$ can be formulated as follows
\begin{align}
    \sup_{a \in \Lc_{2e}} J_\tau (v_a,v_m;\Delta,a).
    \label{J_tau_delta_sup}
\end{align}
It is worth noting that \eqref{J_tau_delta_sup} is introduced to evaluate the worst-case attack impact for each pair of $v_a$ and $v_m$, thus allowing one to compare the impact for different pairs of attack and monitor vertices.
Further,
the worst-case attack impact \eqref{J_tau_delta_sup} is proportional to the alarm threshold $\sigma$ for all possible pairs of $v_a$ and $v_m$.
Therefore, without loss of generality, let us set the alarm threshold $\sigma = 1$ in the remainder of this paper.
\begin{rem}
     Due to the random variable $\Delta \in \Omega$, the worst-case impact \eqref{J_tau_delta_sup} becomes a random variable.
     Thus, in order to compare the worst-case impacts made by pairs of $v_a$ and $v_m$ over the uncertainty set $\Omega$, we need to employ a risk metric which will be introduced in the rest of this subsection.
\end{rem}

After investigating the worst-case attack impact \eqref{J_tau_delta_sup} on the target vertex $v_\tau$ with all the possible pairs of attack $v_a$ and monitor vertices $v_m$, the adversary firstly chooses the attack vertex $v_a$ such that the corresponding risk (defined in \textit{Definition \ref{def:var}}) is maximized \cite{anand2021risk}.
Then, the adversary directly injects the stealthy data injection attack on the input of the selected attack vertex $v_a$.
To this end, the adversary  deals with the following optimization problem:
\begin{align}
	\max_{v_a\neq v_\tau \in \Vc} &~\Jc_\tau(v_a,v_m), 
	\label{J_tau_va}
	\\
	\Jc_\tau(v_a,v_m) = \Rc_{\Delta \in \Omega} &\Big[ ~\sup_{a \in \Lc_{2e}} J_\tau(v_a,v_m;\Delta,a) \Big], \label{game_payoff}
\end{align}
where $\Rc_{\Delta \in \Omega}$ is a risk metric evaluated over the probabilistic uncertainty set. In this paper, we use the well-known Value-at-Risk \cite{duffie1997overview} as our risk metric, which is defined below.
\begin{definition}
    \label{def:var}
    (Value-at-Risk (VaR)): Given a random variable $X$ and $\beta \in (0,1)$, the VaR is defined as 
        \begin{align}
            \text{VaR}_\beta (X) \triangleq 
            \inf \big\{ x | \Pbb 
            \big[ X \leq x \big] \geq 1 - \beta \big\}.
        \end{align}
    With a specified level $\beta \in (0,1)$, VaR$_\beta$ is the lowest amount of $x$ such that with probability $1-\beta$, the random variable $X$ does not exceed $x$.
    \QET
\end{definition}

In order to counter the adversary, the detector adopts the game-theoretic approach to design its detection strategy, which will be introduced in the next part.

\subsection{Game-theoretic approach to sensor placement}
The detector chooses a vertex $v_m \in \Vc \setminus \{v_\tau\}$ and monitors its output with the purpose of minimizing the risk \eqref{game_payoff}. 
Hence, the detector addresses the following problem.
\begin{problem} \label{prob:opt}
	(Optimal monitor selection) Given a target vertex $v_{\tau}$ and an arbitrary attack vertex $v_a$,
	select a monitor vertex that minimizes the risk corresponding to the worst-case attack impact $\Jc_\tau(v_a,v_m)$ defined in \eqref{game_payoff}.
\end{problem}

The above detector objective is converted into the following optimization problem:
\begin{align}
	\min_{v_m \neq v_{\tau} \in \Vc} \Jc_\tau(v_a,v_m). \label{J_tau_vm}
\end{align}

From the scenario of a single-adversary-single-detector we are considering, the adversary objective \eqref{J_tau_va}, and the detector objective \eqref{J_tau_vm}, we formulate Problem~\ref{prob:opt} as a zero-sum game with the game payoff  \eqref{game_payoff} between two players, i.e., the adversary and the detector, as follows:
\begin{align}
    \min_{v_m \neq v_{\tau} \in \Vc} ~
    \max_{v_a \neq v_{\tau} \in \Vc} ~
    \Jc_\tau(v_a,v_m).
    \label{zero-sum_game}
\end{align}
The  min-max optimization problem \eqref{zero-sum_game} admits a saddle-point equilibrium $(v_a^\star,v_m^\star)$ \cite{zhu2015game} if and only if it satisfies 
\begin{align}
    -\infty < \Jc_\tau(v_a,v_m^\star) \leq &\Jc_\tau(v_a^\star,v_m^\star) \leq \Jc_\tau(v_a^\star,v_m) < \infty, ~
    \\
    &~~~~~~~~~~
    \forall v_a,v_m \in \Vc \setminus \{v_\tau\}.
    \label{sadde-point}
\end{align}
The game payoff of the saddle-point equilibrium $\Jc_\tau(v_a^\star,v_m^\star)$ implies that a deviation of the attack vertex $\forall v_a \in \Vc \setminus \{v_\tau,v_a^\star\}$ does not gain the game payoff and a deviation of the monitor vertex $\forall v_m \in \Vc \setminus \{v_\tau,v_m^\star\}$ does not decrease the game payoff.
\begin{rem}
     Since the zero-sum game \eqref{zero-sum_game} determined by discrete decisions of the adversary and the detector might be solved via linear programming \cite[Ch. 5]{boyd2004convex},
     we need to evaluate the game payoff defined in \eqref{game_payoff} for all the possible pairs of $v_a$ and $v_m$. 
     However, computing \eqref{game_payoff} requires us not only to address the non-convexity of the worst-case impact~\eqref{J_tau_delta_sup}
     but also to devise a computationally efficient approximation of \eqref{game_payoff} over a continuous uncertainty set.
\end{rem}

The next section will give us an efficient method to approximately compute the game payoff \eqref{game_payoff} for each selected pair of $v_a$ and $v_m$.

\section{Evaluating the game payoff}
\label{sec:evl_game}
There are two difficulties in solving the zero-sum game \eqref{zero-sum_game}. 
The first difficulty is that: for any given pair of $v_a,v_m$, and uncertainty $\Delta \in \Omega$, the function $\sup_{a \in \Lc_{2e}} J_{\tau}(v_a,v_m;\Delta,a)$ is a non-convex optimization problem. 
Secondly, since the set $\Omega$ is continuous, the problem of assessing the game payoff \eqref{game_payoff} is computationally intractable. 
Thus, in this section, we aim to address both difficulties by invoking the scenario approach \cite{calafiore2007probabilistic} that discretizes the uncertainty set $\Omega$.
\subsection{Worst-case attack impact for a sampled uncertainty point}

We begin by considering the case of a sampled uncertainty realization $\Delta_i \in \Omega$. 
Let us denote the value of the corresponding uncertain Laplacian matrix in \eqref{sys:x_unc} as $L^{\Delta_i}$
and the uncertain system \eqref{sys:x_unc}-\eqref{sys:ym} as
$\Sigma^{\Delta_i}_{\tau,m} \triangleq \big(-L^{\Delta_i},e_a,[e_\tau,e_m]^\top,0\big)$ with the attack input at vertex $v_a$, the target output at vertex $v_\tau$, and the monitor output at vertex $v_m$.
For such an isolated uncertainty, the worst-case attack impact can be written as
\begin{equation}\label{policy_discrete}
    \sup_{a \in \Lc_{2e}} J_{\tau}(v_a,v_m;\Delta_i,a)
\end{equation}

Following the details in  \cite[Prop. 1]{teixeira2021security}, the optimal control problem \eqref{policy_discrete} can be equivalently rewritten as the following convex SDP
\begin{align}
	\gamma_i^\star \triangleq  \underset{\gamma_i \in \Rbb_+, P_i = P_i^\top \geq  0}{\min} & ~~~~ \gamma_i
	\label{opti_LMI_discrete} \\
	\text{s.t.}  ~~~~~~ 
	&R \big( \Sigma^{\Delta_i}_{\tau,m}, P_i, \gamma_i
	\big) \leq 0,  
	\non  
\end{align}
where
\begin{align}
   R \big( \Sigma^{\Delta_i}_{\tau,m}, P_i, \gamma_i
	\big) \triangleq & 
	\ba{cc}
	-L^{\Delta_i} P_i - P_i L^{\Delta_i} ~&~ P_i e_a
	\\ 
	e_a^\top P_i ~&~ 0
	\ea 
	\non \\
	&- 
	\ba{cc}
	\gamma_i  e_m e^\top_m + e_\tau e^\top_\tau ~&~ 0 
	\\ 
	0 ~&~ 0
	\ea.
\end{align}

Next, we tackle the game payoff evaluation over a continuous set of uncertainties $\Omega$ by first approximating the continuous uncertainty set $\Omega$ with a discrete set $\Omega_{M_1}$ of sampled uncertainty realizations, with cardinality $M_1$, and then using the point-wise evaluation of the worst-case attack impact described in~\eqref{opti_LMI_discrete}.

\subsection{Approximate game payoff function}

The game payoff \eqref{game_payoff} is difficult to determine since the risk metric operates over a continuous set $\Omega$. 
To this end, we adopt the scenario approach \cite{calafiore2007probabilistic} to approximate the continuous uncertainty set $\Omega$, and consequently determine the approximate game payoff \eqref{game_payoff}. Before this, we rewrite \eqref{game_payoff} for a given $\beta \in (0,1)$ as \eqref{game_payoff_rewrite}.
\begin{align}\label{game_payoff_rewrite}
    \Jc_{\tau}(v_a,v_m) = ~~ &\inf \gamma\\
      &~ \text{s.t.} ~~  \mathbb{P}_{\Omega} [X \leq \gamma] \geq 1-\beta
\end{align}
where $X = \sup_{a \in \Lc_{2e}} J_{\tau}(v_a,v_m;\Delta,a),~ \Delta \in \Omega$, and the subscript to the probability operator denotes that it operates over the set $\Omega$. 
Next, we apply the scenario approach to determine the approximate value of the optimization problem \eqref{game_payoff_rewrite} in the following theorem. 

\begin{theorem}\label{Th:gamma_N1}
Let $\epsilon_1 \in (0,1)$ represent the accuracy with which the probability operator $\mathbb{P}_{\Omega}$ in \eqref{game_payoff_rewrite} is approximated. Let $\beta_1 \in (0,1)$ represent the confidence with which the accuracy $\epsilon_1$ is guaranteed, i.e.,
\begin{equation}\label{certificates}
    \mathbb{P}\{ \vert \mathbb{P}_{\Omega}(X \leq \gamma) - \hat{\mathbb{P}}_{M_1} \vert \geq\epsilon_1\} \leq \beta_1.
\end{equation}
Here $\hat{\mathbb{P}}_{M_1}$ represents the approximation of the probability operator $\mathbb{P}_{\Omega}$ in \eqref{game_payoff_rewrite} defined as
\begin{equation}
    \hat{\mathbb{P}}_{M_1} \triangleq \frac{1}{M_1}\sum_{i=1}^{M_1}\mathbb{I}\left( X \leq \gamma \right), \; \text{where} \; M_1 \geq \frac{1}{2\epsilon_1^2}\text{log}\frac{2}{\beta_1}.
    \label{th:M1}
\end{equation}
Then, the $\text{VaR}_{\beta}$ defined in \eqref{game_payoff} can be obtained with an accuracy $\epsilon_1$ and confidence $\beta_1$ by solving
\begin{equation}\label{risk_approximate}
\hat{\gamma} = \left\{\begin{aligned}
\min & \quad \gamma\\
\textrm{s.t.} & \quad  \frac{1}{M_1} \sum_{i=1}^{M_1} \mathbb{I}\left( \gamma_i^\star \leq \gamma \right) \geq 1-\beta_1
\end{aligned}\right\},
\end{equation}
where $\hat{\gamma}$ represents the $\text{VaR}_{\beta}$  with an accuracy $\epsilon_1$. 
The value of  $\gamma^\star_i, i \in \{ 1,2,\dots,M_1 \}$, is obtained by solving \eqref{opti_LMI_discrete}.
\QET
\end{theorem}
\begin{proof}
The proof follows directly from our previous results in \cite[Th. 4.4]{anand2021risk}.
\end{proof}
\begin{rem}
     Solving \eqref{risk_approximate} with the risk metric defined in \textit{Definition \ref{def:var}} gives us a measure of risk for a corresponding pair of $v_a$ and $v_m$ that has been evaluated over the explicit probabilistic uncertainty set $\Omega$.
     This risk measure is different from the worst-case impact \eqref{J_tau_delta_sup}, which is a function of a random variable $\Delta \in \Omega$.
\end{rem}
\textit{Theorem~\ref{Th:gamma_N1}} provides a method to compute the approximate value of game payoff \eqref{game_payoff} which was difficult to compute previously. 
In order to evaluate the result of \textit{Theorem~\ref{Th:gamma_N1}},
the next subsection will address the feasibility of the optimization problem \eqref{risk_approximate}.

\subsection{Feasibility analysis}
For $M_1$ sampled uncertainty $\Delta_i \in \Omega_{M_1}$, the following lemma gives us the necessary and sufficient condition to ensure that the  problem \eqref{risk_approximate} is feasible and therefore admits a finite upper bound.
\begin{lemma}[Boundedness]\label{bound_OA_multi}
Consider $M_1$ i.i.d. realizations of uncertainty $\Delta_i \in \Omega_{M_1}$. 
The optimal solution of \eqref{risk_approximate} with these $M_1$ realizations of uncertainty is bounded if and only if the optimal value of \eqref{opti_LMI_discrete} is bounded for at least $\left \lceil{M_1(1-\beta_1)}\right \rceil $ system realizations. 
\QET
\end{lemma}
\begin{proof}
The proof follows directly from our previous results in \cite[Lem. 4.5]{anand2021risk}.
\end{proof}

Then, we investigate the feasibility of the optimization problem \eqref{opti_LMI_discrete} for a system realization corresponding to a given sampled uncertainty $\Delta_i \in \Omega_{M_1}$.
Let us denote the following  systems 
$\Sigma^{\Delta_i}_\tau \triangleq (-L^{\Delta_i},e_a, e^\top_\tau,0)$ and $\Sigma^{\Delta_i}_m \triangleq (-L^{\Delta_i},e_a,e^\top_m,0)$. 
Inspired by \cite[Th. 2]{teixeira2015strategic}, the feasibility of the optimization problem \eqref{opti_LMI_discrete} is related to the invariant zeros of $\Sigma^{\Delta_i}_\tau$ and $\Sigma^{\Delta_i}_m$, which are defined as follows.

\begin{definition}  \label{def:invariant_zero}
	(Invariant zeros)
	Consider the strictly proper system $\Sigma \triangleq (A,B,C,0)$ with $A,B,$ and $C$ are real matrices with appropriate dimensions. A tuple $(\lambda,\bar{x},g) \in \Cbb \times \Cbb^N \times \Cbb$ is a zero dynamics of $\Sigma$ if it satisfies
	\begin{align}
		\ba{cc}
		\lambda I - A ~~~~ & -B \\
		C & 0
		\ea
		\ba{c}
		\bar{x} \\ g
		\ea
		=
		\ba{c}
		0 \\ 0
		\ea,
		~~~ \bar{x} \neq 0.
		\label{definv:mtr_pen}
	\end{align}
	In this case, a finite $\lambda$ is called a finite invariant zero of $\Sigma$.
	Further, the strictly proper system $\Sigma$ always has at least one invariant zero at infinity \cite[Ch. 3]{franklin2002feedback}.
	\QET
\end{definition}

More specifically, let us state the following lemma.
\begin{lemma} \label{lemma:opt_feasible} \cite[Th. 2]{teixeira2015strategic}
    Consider the two following continuous time systems $\Sigma^{\Delta_i}_\tau \triangleq (-L^{\Delta_i},e_a, e^\top_\tau,0)$ and $\Sigma^{\Delta_i}_m \triangleq (-L^{\Delta_i},e_a,e^\top_m,0)$.
    The optimization problem \eqref{opti_LMI_discrete} is feasible if and only if the unstable invariant zeros of $\Sigma^{\Delta_i}_m$ are also invariant zeros of $\Sigma^{\Delta_i}_\tau$.
    \QET
\end{lemma}

Inspired by \textit{Lemma~\ref{lemma:opt_feasible}},
we will investigate both finite and infinite invariant zeros of the two systems $\Sigma^{\Delta_i}_m$ and $\Sigma^{\Delta_i}_\tau$.

\subsubsection*{\textbf{Finite invariant zeros}}
Let us state the
following lemma that considers the finite invariant zeros of $\Sigma^{\Delta_i}_m$.
\begin{lemma}  \label{lemma:no_un_zero}
	Consider a networked control system associated with a connected undirected graph $\Gc \triangleq (\Vc,\Ec,A,\Theta)$, 
	whose closed-loop dynamics is described in \eqref{sys:x_unc}-\eqref{sys:ym} for a given sampled uncertainty $\Delta_i \in \Omega_{M_1}$.
	Suppose that the networked control system is driven by the stealthy data injection attack  at a single attack vertex $v_a$, and observed by a single monitor vertex $v_m$, resulting in the state-space model $\Sigma^{\Delta_i}_m \triangleq (-L^{\Delta_i},e_a,e^\top_m,0)$.
	Then, there exist self-loop control gains $\theta_i, ~ i \in \{1,2,\ldots,N\},$ in \eqref{sys:u} such that the networked control system $\Sigma^{\Delta_i}_m$ has no 
	finite unstable invariant zero.
	\QET
\end{lemma}
\begin{proof}
    We postpone the proof to Appendix A.
\end{proof}

The constructive proof of \textit{Lemma~\ref{lemma:no_un_zero}} (see Appendix A) gives us a design procedure to ensure that the system $\Sigma^{\Delta_i}_m$ has no finite unstable zero.

\subsubsection*{\textbf{Infinite invariant zeros}}
We now investigate the infinite invariant zeros of the systems $\Sigma^{\Delta_i}_m$ and $\Sigma^{\Delta_i}_\tau$. 
In the investigation, we make use of known results connecting infinite invariant zeros mentioned in \textit{Definition \ref{def:invariant_zero}} and the relative degree of a linear system, which is defined below.
\begin{definition} \label{def:rela_deg}
	(Relative degree)
	\cite[Ch. 13]{khalil2002nonlinear} Consider the strictly proper system $\Sigma \triangleq (A,B,C,0)$ with $A \in \Rbb^{n \times n}$, $B$, and $C$ are real matrices with appropriate dimensions.
	The system $\Sigma$ is said to have relative degree $r ~ (1 \leq r \leq n) $ if the following conditions satisfy
	\begin{align}
		&C A^{k} B = 0, ~~ 0 \leq k < r-1,
		\non \\
		&C A^{r-1} B \neq 0.
		\label{def_red}
	\end{align}
\end{definition}

Based on \textit{Definition~\ref{def:rela_deg}}, let us denote $r_{\tau a}$ and $r_{ma}$ as the relative degrees of $\Sigma^{\Delta_i}_\tau$ and $\Sigma^{\Delta_i}_m$, respectively.
In the scope of this study, we have assumed that the cyber-attack \eqref{sys:ua} has no direct impact on the outputs \eqref{sys:yt} and \eqref{sys:ym}, resulting in strictly proper systems
$\Sigma^{\Delta_i}_\tau$ and $\Sigma^{\Delta_i}_m,~\forall \Delta_i \in \Omega_{M_1}$.
This implies that the relative degrees $r_{\tau a}$ and $r_{ma}$ of $\Sigma^{\Delta_i}_\tau$ and $\Sigma^{\Delta_i}_m$ are positive, yielding their infinite zeros.
By following our existing result related to those infinite zeros \cite[Th. 7]{nguyen2022single}
the infinite zeros of $\Sigma_m^{\Delta_i}$ are also the infinite zeros of $\Sigma_\tau^{\Delta_i}$ if and only if the following condition holds
\begin{align}
	r_{ma} \leq r_{\tau a}. 
	\label{cond_red}
\end{align}

\subsubsection*{\textbf{Boundedness of solutions}} After analyzing both finite and infinite zeros of the two systems $\Sigma_m^{\Delta_i}$ and $\Sigma_\tau^{\Delta_i}$, the following theorem gives us a sufficient condition to ensure the feasibility of the optimization problem \eqref{opti_LMI_discrete}, and thus of the existence of a finite upper bound on the corresponding optimal value.
\begin{theorem}
    \label{th:opt_fea}
    Consider the strictly proper systems $\Sigma^{\Delta_i}_\tau \triangleq (-L^{\Delta_i},e_a,e^\top_\tau,0)$ and $\Sigma^{\Delta_i}_m \triangleq (-L^{\Delta_i},e_a,e^\top_m,0)$, in which the two systems have the same stealthy data injection attack input at a single attack vertex $v_a$ but different output vertices, i.e., $v_\tau$ for $\Sigma^{\Delta_i}_\tau$ and $v_m$ for $\Sigma^{\Delta_i}_m$.
	Suppose the systems $\Sigma^{\Delta_i}_\tau$ and $\Sigma^{\Delta_i}_m$ have relative degrees $r_{\tau a}$ and $r_{ma}$, respectively.
	Then, the  problem \eqref{opti_LMI_discrete} admits a finite solution if 
	\begin{enumerate}
	    \item the self-loop control gains $\theta_i,~i \in \{1,2,\ldots,N\}$, in \eqref{sys:u} are chosen such that the system $\Sigma^{\Delta_i}_m$ has no finite unstable zeros; and
	    \item the condition \eqref{cond_red} holds.
	    \QET
	\end{enumerate}
\end{theorem}
\begin{proof}
    The proof is postponed to Appendix B.
\end{proof}

The sufficient condition \eqref{cond_red} will be verified by computing the approximate game payoffs \eqref{risk_approximate} in \textit{Theorem~\ref{Th:gamma_N1}} and the equilibrium of the zero-sum game \eqref{zero-sum_game} will be analyzed via a numerical example in the next section.

\section{Numerical examples}
\label{sec:num_eg}
To validate the obtained results, through a numerical example, this section 
$i)$ applies \eqref{risk_approximate} with two different values of $\beta$ to the example with the aim of verifying \eqref{cond_red}; 
$ii)$ examines the saddle-point equilibrium of the zero-sum game \eqref{zero-sum_game} with the two different values of $\beta$;
$iii)$ computes the mixed-strategy Nash equilibrium of the zero-sum game in case there is no saddle-point equilibrium.
Let us take an example of a 10-vertex  networked control system depicted in Fig. \ref{fig:graph_10vertex}. 
The simulation parameters are chosen as follows: 
\begin{align}
    L^\Delta & \triangleq [\ell_{ij}^\Delta] =  [\bar{\ell}_{ij}] + [\delta_{ij}] + \Theta,
    \\
    \bar{\ell}_{ij} &= -10, ~\delta_{ij} \in [-0.5, 0.5], ~ \forall (v_i,v_j) \in \Ec, i \neq j,
    \\
    \bar{\ell}_{ij} &= \delta_{ij} = 0, ~ \forall (v_i,v_j) \notin \Ec,
    \\
    \ell_{ii}^\Delta &= -\sum_{v_j \in \Nc_i} \big( \bar{\ell}_{ij} + \delta_{ij} \big),
    ~\theta_0 = 0.5. 
\end{align}
\subsection{Computing the approximate game payoff}
To compute \eqref{risk_approximate}, let us choose $\epsilon_1 = 0.06, \beta_1 = 0.08$,  
and $M_1 = 450$, which satisfy \eqref{th:M1}.
For any sampled uncertainty $\Delta_i \in \Omega_{M_1}$,
the chosen uniform offset self-loop control gain $\theta_0$ (see Appendix A) ensures that
$\Sigma^{\Delta_i}_m$ has no finite unstable zero, which validates \textit{Lemma~\ref{lemma:no_un_zero}}.
We will present two cases by selecting two values of the
specified level $\beta$ in \eqref{game_payoff_rewrite}, i.e., $\beta_a = 0.08$ and $\beta_b = 0.15$.
Suppose that $v_5$ is the protected target vertex (see \textit{Assumption~\ref{ass:target_vertex}} and Fig.~\ref{fig:graph_10vertex}). 
There are two possible monitor vertices $v_2$ and $v_6$, which satisfy the necessary and sufficient condition \eqref{cond_red} for any $v_a \in \Vc \setminus \{v_5\}$ (see Fig. \ref{fig:graph_10vertex}).
For more clarity, we compute the approximate game payoff \eqref{risk_approximate} w.r.t. the target vertex $v_5$ for each pair of $v_a \in \Vc \setminus \{v_5\}$ and $v_m \in \{v_2,v_6\}$ in the cases $\beta = \beta_a$ and $\beta = \beta_b$, which gives us
\begin{align}
    \Jc_5(v_a,v_{m=2};\beta_a) \leq 1.5848,~
    \Jc_5(v_a,v_{m=6};\beta_a) \leq 1.5055,
    \\
    \Jc_5(v_a,v_{m=2};\beta_b) \leq 1.5550,~
    \Jc_5(v_a,v_{m=6};\beta_b) \leq 1.4803.
\end{align}
Otherwise, there exits  at least an attack vertex $v_a \in \Vc \setminus \{v_5\}$ pairing with an arbitrary monitor $v_m \in \Vc \setminus  \{v_2,v_5,v_6\}$ to yield infinite game payoffs, e.g., $\Jc_5(v_{a=3},v_{m=1};\beta_a) = \infty$,
$\Jc_5(v_{a=3},v_{m=1};\beta_b) = \infty$,
$\Jc_5(v_{a=10},v_{m=3};\beta_a) = \infty$, and
$\Jc_5(v_{a=10},v_{m=3};\beta_b) = \infty$.
In order to explain those infinite values,
we verify the  condition \eqref{cond_red} by checking the relative degrees among those vertices via Fig.~\ref{fig:graph_10vertex}, i.e., 
$r_{a=3,m=5} = 2 < r_{a=3,m=1} = 3$ and $r_{a=10,m=5} = 2 < r_{a=10,m=3} = 3$, which violate the necessary and sufficient condition \eqref{cond_red}.

\begin{figure} [!t]
	\centering
	\includegraphics[width=0.4\textwidth]{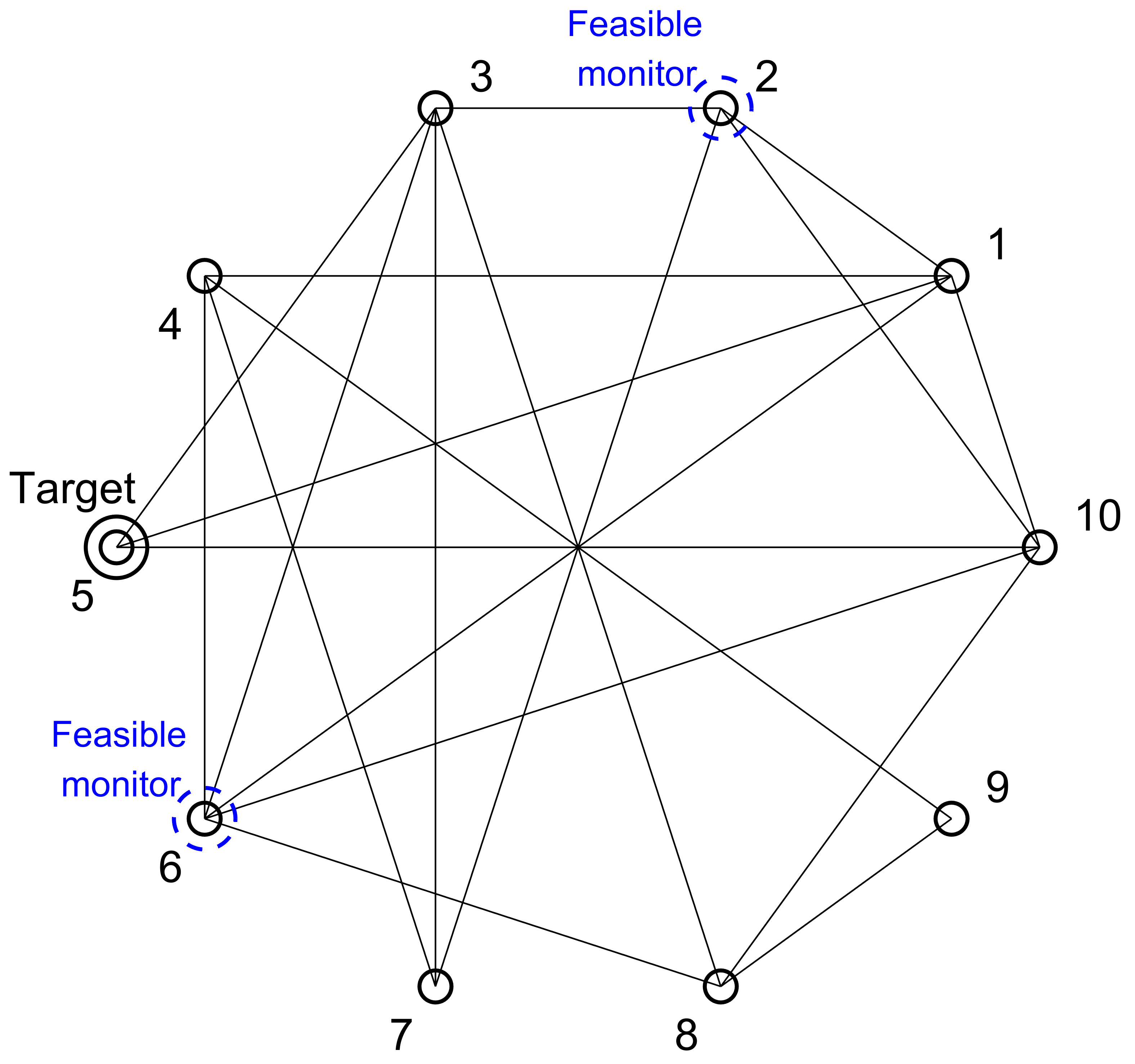}
	\caption{10-vertex networked control system with target vertex $v_5$.}
	\label{fig:graph_10vertex}
	\includegraphics[width=0.4\textwidth]{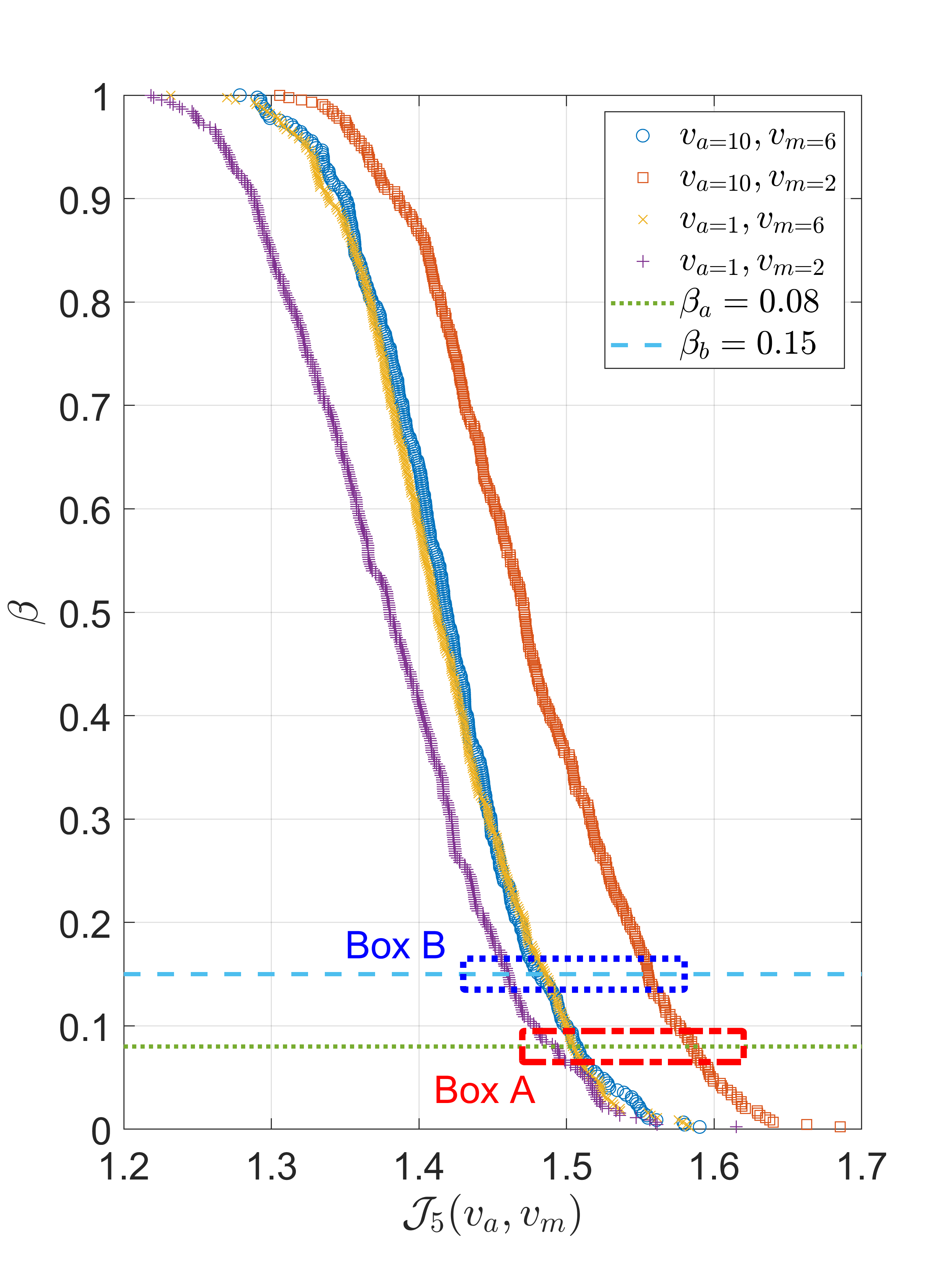}
    \caption{Approximate game payoff \eqref{risk_approximate} with $\beta = \beta_a = 0.08$ and $\beta = \beta_b = 0.15$ in case the detector selects $v_{m=2}$ or $v_{m=6}$ and the adversary attacks $v_{a=1}$ or $v_{a=10}$.
    The other game payoffs yielded by the other choices of $v_a$ and $v_m$ are removed due to the ineffectiveness.
    }
    \label{fig:var_beta8}
\end{figure}
\subsection{Examining the saddle-point equilibrium}
Next, we will investigate the equilibrium of the zero-sum game in the cases $\beta = \beta_a$ and $\beta = \beta_b$. 
Fig.~\ref{fig:var_beta8} illustrates the game payoffs for $v_a \in \{v_1,v_{10}\}$ and $v_m \in \{v_2,v_6\}$ corresponding to $\Delta_i \in \Omega_{M_1}$.
In both cases $\beta = \beta_a$ and $\beta = \beta_b$, since those game payoffs dominate the values of the other choices of $v_a \in \Vc \setminus \{v_1,v_5,v_{10}\}$ and $v_m \in \{v_2,v_6\}$, we only show four marked-lines in Fig.~\ref{fig:var_beta8}.

\subsubsection*{\textbf{In the first case $\beta = \beta_a = 0.08$}}
the crossing points of the green dotted-line and marked-lines are the approximate game payoffs with $\beta = \beta_a$ for the corresponding pairs of attack and monitor vertices (see Box A in Fig.~\ref{fig:var_beta8}).
By observing those approximate game payoffs in Fig.~\ref{fig:var_beta8}, one has
\begin{align}
    &~~~~~~~~~~
    \Jc_5(\forall v_a \in \Vc \setminus \{ v_5,v_{10}\},v_{m=6} ;\beta_a) 
    \\
    &<
    \Jc_5(v_{a=10},v_{m=6};\beta_a) 
    < \Jc_5(v_{a=10},v_{m=2};\beta_a).
    \label{sadde-point-beta1}
\end{align}
According to the definition of the saddle-point equilibrium in \eqref{sadde-point}, the  inequalities \eqref{sadde-point-beta1} imply that the example admits a saddle-point equilibrium $(v_{a}^\star = v_{10},v_m^\star = v_6)$ with $\beta = \beta_a$.
\subsubsection*{\textbf{In the second case $\beta = \beta_b = 0.15$}}
the approximate game payoffs are the crossing points of the marked-lines and the blue dashed-line (see Box B in Fig.~\ref{fig:var_beta8}).
Those crossing points give us \\
\begin{align}
    \Jc_5(v_{a=1},v_{m=2};\beta_b) = 1.4603,
    \\
    \Jc_5(v_{a=10},v_{m=6};\beta_b) = 1.4803,
    \\
    \Jc_5(v_{a=1},v_{m=6};\beta_b) = 1.4856, 
    \\
    \Jc_5(v_{a=10},v_{m=2};\beta_b) = 1.5550. 
    \label{sadde-point-beta2}
\end{align}
From \eqref{sadde-point-beta2}, we will examine whether a saddle-point equilibrium exists.
If the detector monitors $v_{m=2}$, the adversary simply attacks $v_{a=10}$ to maximize the risk.
But, in the case of $v_{a=10}$, the detector can move to $v_{m=6}$ to reduce the risk since $\Jc_5(v_{a=10},v_{m=6};\beta_b) < \Jc_5(v_{a=10},v_{m=2};\beta_b)$.
Then, the adversary can obtain a higher risk by attacking $v_{a=1}$ instead of $v_{a=10}$, i.e., $\Jc_5(v_{a=1},v_{m=6};\beta_b) > \Jc_5(v_{a=10},v_{m=6};\beta_b)$.
Monitoring $v_{m=2}$ yields a lower risk for the detector, i.e., $\Jc_5(v_{a=1},v_{m=2};\beta_b) < \Jc_5(v_{a=1},v_{m=6};\beta_b)$.
The story comes back to the beginning since the adversary simply attacks $v_{a=10}$ to maximize the risk.
The above observation implies that the example with $\beta = \beta_b$ does not admit a saddle-point equilibrium defined in \eqref{sadde-point}.
However, the game
always admits a mixed-strategy Nash equilibrium \cite{zhu2015game}, which will be computed in the next subsection.
\subsection{Computing mixed-strategy Nash equilibrium}
We compute the mixed-strategy Nash equilibrium for the example with the cases $\beta = \beta_a$ and $\beta = \beta_b$.
Let us denote $\Pbb(v_a;\beta)$ and $\Pbb(v_m;\beta)$, $\beta \in \{\beta_a = 0.08,\beta_b =0.15\}$ as the probabilities for attack $v_a$ and monitor vertices $v_m$, respectively. 
For convenience, we denote
$\bar{\Pbb}(v_a;\beta) = \left[\Pbb(v_{a=1};\beta),\ldots,\Pbb(v_{a=10};\beta)\right]^\top, (v_a \neq v_5)$
and
\\
$\bar{\Pbb}(v_m;\beta) = \left[\Pbb(v_{m=1};\beta),\ldots,\Pbb(v_{m=10};\beta)\right]^\top, (v_m \neq v_5)$.
The expected game payoff of the example w.r.t. the target vertex $v_5$ for attack vertex $v_a$ and monitor vertex $v_m$ is 
\begin{align}
	Q_5(v_a,v_m;\beta) = \bar{\Pbb}(v_a;\beta)^\top \bar{\Jc}_5 \bar{\Pbb}(v_m;\beta), 
\end{align}
where $\bar{\Jc}_5 = \big[\Jc_5(v_i,v_j;\beta)_{ij}\big]$ is a $9\times9$-game matrix, whose $ij$-entry is filled by $\Jc_5(v_{a=i},v_{m=j};\beta)$.
Similarly to \eqref{sadde-point},
there exits a saddle point $(v_a^\star,v_m^\star)$ if it satisfies
\begin{align}
	Q_5(v_a,v_m^\star;\beta) \leq &Q_5(v_a^\star,v_m^\star;\beta) \leq Q_5(v_a^\star,v_m;\beta),
	\\ 
	& ~~~~~~~~~~~~~~
	\forall v_a,v_m   \in \Vc \setminus \{v_\tau\}.
	\label{exp_gamepayoff}
\end{align}
The saddle point $(v_a^\star,v_m^\star)$ in \eqref{exp_gamepayoff} indicates that a deviation of selecting $v_a(v_m)$ does not increase(decrease) the optimal expected game payoff $Q_5(v_a^\star,v_m^\star;\beta)$. 
Further,
since the possible choices of the detector are restricted to $\{v_2,v_6\}$, we simply obtain $\Pbb(\forall v_m \in \Vc \setminus \{v_2,v_5,v_6\};\beta) = 0$. 
More specifically, by using linear programming \cite[Ch. 5]{boyd2004convex} to compute \eqref{exp_gamepayoff}, we receive the following optimal solution
\subsubsection*{\textbf{In the first case $\beta = \beta_a = 0.08$}}
\begin{align}
    &\Pbb^\star(v_{m=6};\beta_a) = 100 \%,~ \Pbb^\star(v_{m=2};\beta_a) = 0 \%, \\
    &\Pbb^\star(v_{a=10};\beta_a) = 100 \%, ~ \Pbb^\star(\forall v_a \in \Vc \setminus  \{5,10\};\beta_a) = 0 \%.
\end{align}
The above optimal solution once again confirms that a pair $(v_{a}^\star = v_{10},v_m^\star = v_6)$ is the pure saddle-point equilibrium \eqref{sadde-point} of the example with $\beta = \beta_a$, which was also verified in \eqref{sadde-point-beta1}.

\subsubsection*{\textbf{In the second case $\beta = \beta_b = 0.15$}} we obtain the following optimal solution
\begin{align}
    &\Pbb^\star(v_{m=6};\beta_b) \approx 94.72 \%,~ \Pbb^\star(v_{m=2};\beta_b) \approx 5.28 \%, \\
    &\Pbb^\star(v_{a=10};\beta_b) \approx 25.29 \%, ~ \Pbb^\star(v_{a=1};\beta_b) \approx 74.71 \%,
    \\
    &\Pbb^\star(\forall v_a \in \Vc \setminus \{1,5,10\};\beta_b) = 0 \%.
\end{align}
The above optimal solution clearly show that the example does not admit a pure saddle-point equilibrium \eqref{sadde-point} with $\beta = \beta_b$, which was discussed at the end of the previous subsection.
\section{Conclusion}
\label{sec:concl}
In this paper, we studied a continuous-time networked control system attacked by an adversary with uncertain system knowledge. 
The purpose of the adversary was to manipulate the output of a protected target vertex by directly conducting the stealthy data injection attack on another vertex.
Meanwhile,
an optimal sensor placement problem was formulated such that
a detector with the same uncertain system knowledge places a sensor at a vertex in order to unmask the adversary. 
We developed a risk-based game-theoretic framework to describe the interactions between the two players, the adversary and the detector, in the presence on probabilistic parameter uncertainty. In particular, we formulate the optimal decisions as a zero-sum game, where the game payoff is taken as a risk metric evaluated over the probabilistic uncertainty set.
Due to the continuous nature of the uncertainty set, the zero-sum game could not be solved directly. 
Thus, we employed the scenario approach to approximately compute the game payoff over a number of samples of uncertain parameters.
After approximately evaluating the game payoff for each pair of monitor and attack vertices, the mixed-strategy Nash equilibrium of the zero-sum game was also computed by linear programming.
In future works, our game will be expanded to consider multiple attack and monitor vertices.
Characterizing an analytical solution to the equilibrium of the game between the adversary and the detector would also be a promising topic.

\section*{Appendix}
\subsection{Proof of Lemma~\ref{lemma:no_un_zero}}
\label{app_pf_lemma_opt_fea}
    
    Let us denote a tuple $(\lambda^{\Delta_i}_m,\bar{x}^{\Delta_i}_m,g^{\Delta_i}_m) \in \Cbb \times \Cbb^N \times \Cbb$ as a zero dynamics of $\Sigma_m^{\Delta_i}$, where a finite $\lambda_m^{\Delta_i}$ is called a finite invariant zero of $\Sigma_m^{\Delta_i}$.
    From \textit{Definition~\ref{def:invariant_zero}}, one has that the tuple $(\lambda^{\Delta_i}_m,\bar{x}^{\Delta_i}_m,g^{\Delta_i}_m)$ satisfies
    \begin{align}
         \ba{cc}
        \lambda^{\Delta_i}_m  I + L^{\Delta_i}  & -e_a \\
        e_m^\top & 0
        \ea 
        \ba{c} \bar{x}^{\Delta_i}_m \\ g^{\Delta_i}_m \ea 
        = 
        \ba{c} 0 \\ 0 \ea.
        \label{pen_mtr_lam_m}
    \end{align}
    The above equation is rewritten as
    \begin{align}
        \ba{cc}
        (\lambda^{\Delta_i}_m - \theta_0) I + L^{\Delta_i} + \theta_0 I & -e_a \\
        e_m^\top & 0
        \ea 
        \ba{c} \bar{x}^{\Delta_i}_m \\ g_m \ea 
        = 
        \ba{c} 0 \\ 0 \ea,
        \\
        \label{pen_mtr_lam_m1}
    \end{align}
    where $\theta_0 \in \Rbb_+$ is a uniform offset self-loop control gain.
    From \eqref{pen_mtr_lam_m1}, the finite value $(\lambda^{\Delta_i}_m-\theta_0) \in \Cbb$ is an invariant zero of a new state-space model $\Sigma_{0m}^{\Delta_i} \triangleq (-L^{\Delta_i}-\theta_0I,e_a,e^\top_m,0)$.
    For all $\lambda_m^{\Delta_i} \in \Cbb$ satisfies \eqref{pen_mtr_lam_m1},
    the control gain $\theta_0$ can be adjusted such that $\theta_0 >$ Re$(\lambda_m^{\Delta_i})$, resulting in that $\Sigma_{0m}^{\Delta_i}$ has no finite unstable zero.
    Then, the self-loop control gains $\theta_i, ~ i \in \{1,2,\ldots,N\},$ in \eqref{sys:u} are tuned with $\theta_0$ such that the system $\Sigma_m^{\Delta_i}$  is identical with $\Sigma_{0m}^{\Delta_i}$.
    By this tuning procedure, the system $\Sigma^{\Delta_i}_{m}$ also has no finite unstable invariant zero.
\QEDB
\subsection{Proof of Theorem~\ref{th:opt_fea}}
\label{app_pf_th_opt_fea}
	Based on \textit{Lemma~\ref{lemma:opt_feasible}},
	the optimization problem \eqref{opti_LMI_discrete} is feasible if and only if  $\Sigma^{\Delta_i}_m$ has  unstable invariant zeros that are also invariant zeros of  $\Sigma^{\Delta_i}_\tau$.
	By applying the control design procedure in the proof of \textit{Lemma \ref{lemma:no_un_zero}} (see Appendix A), we ensure that $\Sigma^{\Delta_i}_m$ has no finite unstable invariant zeros, which leaves us to analyze infinite zeros of those systems.
    Recall the equivalence between the relative degree of a SISO system and the degree of its infinite zero. 
 	Hence, a necessary condition to guarantee the feasibility of the optimization \eqref{opti_LMI_discrete} is that the number of infinite invariant zeros of $\Sigma^{\Delta_i}_m$ is not greater than that of $\Sigma^{\Delta_i}_\tau$.
 	This implies  $r_{m a} \leq r_{\tau a}$. 
 	For sufficiency, it remains to show that 
 	if $r_{m a} \leq r_{\tau a}$, any infinite zeros of $\Sigma^{\Delta_i}_{m}$ are also infinite zeros of $\Sigma^{\Delta_i}_{\tau}$.
The proof directly follows our previous results \cite[Th. 7]{nguyen2022single}.
\QEDB


\bibliographystyle{ieeetr} 
\bibliography{ref}             


\end{document}